\documentclass[twoside]{article}

\usepackage[tbtags]{amsmath}
\usepackage{amsthm}
\usepackage{amsfonts,amssymb,mathrsfs}

\usepackage{graphicx}
\usepackage{esint}

\theoremstyle{plain}
\newtheorem{theorem}{Theorem}
\newtheorem{lemma}{Lemma}
\newtheorem{corollary}{Corollary}

\theoremstyle{definition}

\newtheorem{remark}{Remark}

\begin{document}

\title{On one real basis for $L^{2}\left(\mathbb{Q}_{p}\right)$}

\author{A.\,Kh.~Bikulov, \thanks{Institute of Chemical Physics, Kosygina Street 4, 117734 Moscow, Russia, e-mail: bikulov1903@rambler.ru}
\\
 A.\,P.~Zubarev \thanks{Samara State Aerospace University, Samara, Russia and Samara State Transport University, Samara, Russia, e-mail: apzubarev@mail.ru}}
\maketitle
\begin{abstract}
We construct new bases of real functions from $L^{2}\left(B_{r}\right)$
and from $L^{2}\left(\mathbb{Q}_{p}\right)$. These functions are eigen\-functions of the $p$-adic pseudo-differential Vladimirov operator,
which is defined on a~compact set $B_{r}\subset\mathbb{Q}_{p}$ of the field of $p$-adic
numbers $\mathbb{Q}_{p}$ or, respectively, on the entire field $\mathbb{Q}_{p}$.
A~relation between the basis of  functions from $L^{2}\left(\mathbb{Q}_{p}\right)$
and the  basis of  $p$-adic wavelets from $L^{2}\left(\mathbb{Q}_{p}\right)$ is found.
As an application, we consider the solution of the Cauchy problem with the initial condition
on a~compact set for a~pseudo-differential equation with a~general pseudo-differential operator, which is diagonal in the basis constructed.

\end{abstract}

\section{Introduction}

The theory of $p$-adic pseudo-differential operators was developed by V.\,S.~Vladimirov~\cite{V}  (see also \cite{VVZ,DKKV} and the references given
therein). The operator $D^{\alpha}$, which is known in the literature on $p$-adic mathematical physics, was defined in~\cite{V} as a~pseudo-differential Vladimirov operator $D^{\alpha}$.
At present the pseudo-differential equations involving Vladimirov operator are widely useful in various branches of mathematics and physics: $p$-adic
evolution equations and path integrals \cite{S-Sh-1,S-Sh-2},
$p$-adic-valued random processes \cite{VVZ,Koch,S-H}, $p$-adic modeling of conformational dynamics of proteins \cite{ABKO,ABO,ABZ,ABZ_MIAN},
$p$-adic modeling of prototypes of molecular nano-machines \cite{ABZ_motor},
$p$-adic models of economic and social systems~\cite{BZK}, etc.

The Vladimirov operator is diagonalized by the $p$-adic Fourier transform.
There are also examples of bases consisting of compactly supported eigenfunctions of the Vladimirov operator.
Examples of bases of eigenfunctions of the Vladimirov operator were constructed in \cite{Koch,V1,V2,V3}.
In \cite{Koz} there is another example of an orthonormal basis consisting of eigenvectors
of the Vladimirov operator (which was called the basis of $p$-adic wavelets).

The Vladimirov operator and similar operators with translation invariant kernels that are diagonalized by the Fourier transform are also diagonalized by wavelets,
but a~wavelet basis is capable of diagonalizing even a~more general construction of pseudo-differential operators that cannot be diagonalized by the Fourier transform \cite{Koz,Koz1,KK}.

In the beginning of в Section~2 we outline the Fourier analysis on compact sets. This analysis was applied several times by the authors, but was never mathematically formalized.
We think that this machinery, which is widely applied in physical applications, deserves a~separate systematic treatment.
An expansion in a~discrete basis of characters on a~compact set is nothing else but the analogue of the discrete Fourier transform.
This basis of characters is an eigenbasis for the modified Vladimirov operator on a~compact set  $B_{r}\subset Q_{p}$, which has the following form:
\[
D^{\alpha}(B_{r})f(x)=\intop_{B_{r}}\frac{f(y)-f(x)}{|x-y|_{p}^{\alpha+1}}dy.
\]
This operator is distinct from the restriction of the Vladimirov operator to $B_{r}$,
which was considered in~\cite{VVZ}. This operator is more natural in applications of
the $p$-adic analysis to physics, and in particular, to modeling of conformational dynamics of protein
\cite{ABKO,ABO,ABZ,ABZ_MIAN,ABZ_motor}. As was pointed out above, expansions of functions in the basis under consideration
were used by the authors of the present paper in a~number of studies.
For a~first time such an expansion was constructed on~$\mathbb{Z}_{p}$ in~\cite{Bik_Vol}. In~\cite{ABKO},
a~solution of a~system of equations of the  `reaction-diffusion' type on~$B_{r}$ was put forward in terms of
an expansion in such a~basis. Further, in Section~3 we construct a~new real basis for $L^{2}\left(B_{r}\right)$. This
basis is a~complete family of eigenfunctions of the Vladimirov operator on~$B_{r}$ (and is also a~complete family of
eigenfunctions for more general operators). In Section~4 we consider an extension of this basis to a~basis for $L^{2}\left(\mathbb{Q}_{p}\right)$.
We shall be also concerned with the relation between this basis and the basis of $p$-adic
wavelets. As an application, we consider the solution of the Cauchy problem with the initial condition on a~compact set for pseudo-differential
equation with pseudo-differential operator of a~general form, which is diagonal in this basis.

\section{Fourier transform, convolution and spectrum of the Vladimirov operator on~$B_{r}$}

Let $\mathbb{Q}$ be the field of rational numbers and let $p$ be a~fixed prime number. Any rational number  $x\ne0$ is uniquely representable in the form
\begin{equation}
x=p^{\gamma}\frac{a}{b},\label{x}
\end{equation}
where $a, b, \gamma\in\mathbb{Z}$ are integers, $a$ and $b$ are natural numbers relative prime to~$p$ and having no common divisors.
The $p$-adic norm  $\left|x\right|_{p}$ of a~number $x\in\mathbb{Q}$ is defined by  $\left|x\right|_{p}=p^{-\gamma}$
, $\left|0\right|_{p}=0$. The completion of the field of rational numbers  $\mathbb{Q}$ with respect to the $p$-adic norm forms the field of $p$-adic
numbers $\mathbb{Q}_{p}$. Endowed with the metric $d(x,y)=\left|x-y\right|_{p}$ the field $\mathbb{Q}_{p}$ becomes a~complete compact separable extremally disconnected
ultrametric  space.

We shall denote by $B_{r}(a)=\{x\in\mathbb{Q}_{p}:\:|x-a|_{p}\leq p^{r}\}$
the ball of radius  $p^{r}$ centered at~$a$, denote by $S_{r}(a)=\{x\in\mathbb{Q}_{p}:\:|x-a|_{p}=p^{r}\}$
the sphere of radius $p^{r}$ with centre at~$a$, $B_{r}\equiv B_{r}(0)$,
$S_{r}\equiv S_{r}(a)$, $\mathbb{Z}_{p}\equiv B_{0}$. We shall also need the following functions
\[
\Omega(\left|x\right|_{p}p^{\gamma})=\left\{ \begin{array}{l}
1,\mathrm{\;}\left|x\right|_{p}p^{\gamma}\leq1,\\
0,\mathrm{\;}\left|x\right|_{p}p^{\gamma}>1,
\end{array}\right.\qquad
\delta(\left|x\right|_{p}-p^{\gamma})=\left\{ \begin{array}{l}
1,\mathrm{\;}\left|x\right|_{p}=p^{\gamma},\\
0,\mathrm{\;}\left|x\right|_{p}\neq p^{\gamma}.
\end{array}\right.
\]
On $\mathbb{Q}_{p}$ there exists a~unique (up to a~factor)
Haar measure $d_{p}x$ which is invariant with respect to translations $d_{p}\left(x+a\right)=d_{p}x$.
We assume that $d_px$ is a~full measure; that is,
\begin{equation}
\intop_{\mathbb{Z}_{p}}d_{p}x=1.\label{norm}
\end{equation}
Under this hypothesis the measure $d_{p}x$ is unique.

We now introduce the class $W_{l}^{\alpha}$ $(\alpha\geq0)$ of complex
functions $f(x)$ on~$\mathbb{Q}_{p}$ satisfying the following conditions:

1) $\left|\varphi(x)\right|\le C\left(1+\left|x\right|_{p}^{\alpha}\right)$,
where $C$ is a~real positive number;

2) there exists a~natural number $l$ such that $\varphi\left(x+x'\right)=\varphi\left(x\right)$
for any $x\in\mathbb{Q}_{p}$ and any $x'\in\mathbb{Q}_{p}$,  $\left|x'\right|_{p}\le p^{l}$.
A~function $\varphi(x)$ satisfying such condition is called locally constant, and the number~$l$ is called the exponent of local constancy of a~function.

Let $D$ be the class of compactly supported functions from $W^{0}$, which are called test functions
(Bruhat--Schwartz functions), and let  $D'$ be the class of all generalized functions on~$D$. We shall also denote by
$D_{r}^{l}$ the class of functions from $W_{l}^{0}$ supported in the ball  $B_{r}=\left\{ x:\:|x|_{p}\leq p^{r}\right\} $.

Next, let $\chi$ be the normalized additive character of the field~$\mathbb{Q}_{p}$.
We have $\chi\in W^{0}$. The Fourier transform of a~function $\varphi\left(x\right)\in L^{1}\left(\mathbb{Q}_{p},d_{p}x\right)$
is defined by
\[
\tilde{\varphi}\left(k\right)=F\left[\varphi\left(x\right)\right](k)\equiv\intop_{\mathbb{Q}_{p}}\chi\left(kx\right)\varphi\left(x\right)d_{p}x,\;k\in\mathbb{Q}_{p}.
\]
For a $\tilde{\varphi}\left(k\right)\in L^{1}\left(\mathbb{Q}_{p},d_{p}k\right)$, the inverse transform is defined by
\[
\varphi\left(x\right)=F^{-1}\left[\tilde{\varphi}\left(k\right)\right](x)\equiv\intop_{\mathbb{Q}_{p}}\chi\left(-kx\right)\tilde{\varphi}\left(k\right)d_{p}k,\;x\in\mathbb{Q}_{p}.
\]

We have the following result.

\begin{theorem}
If $f(x)\in D_{r}^{l}(\mathbb{Q}_{p})$, then  $\tilde{f}(k)\in D_{-l}^{-r}(\mathbb{Q}_{p})$.
\label{th1}
\end{theorem}

The operator $D^{\alpha}$ (the pseudo-differential Vladimirov operator
\cite{VVZ}), $\alpha>0$, acts on  functions $f\in W_{l}^{\beta}\left(\mathbb{Q}_{p}\right)$, $0\leq\beta<\alpha$, by the formula
\begin{equation}
D_{x}^{\alpha}f\left(x\right)=-\frac{1}{\Gamma\left(-\alpha\right)}\intop_{\mathbb{Q}_{p}}d_{p}y\frac{f\left(y\right)-f\left(x\right)}{\left\vert x-y\right\vert _{p}^{\alpha+1}},\label{Vlad_Q_p}
\end{equation}
where
\[
\Gamma_{p}(-\alpha)=\frac{1-p^{-\alpha-1}}{1-p^{\alpha}}
\]
is the $p$-adic analogue of the gamma-function.

We shall also define the analogue of operator (\ref{Vlad_Q_p}) on $W^{0}\left(B_{r}\right)$ by the formula
\begin{equation}
D^{\alpha}(B_{r})\varphi(x)\equiv-\frac{1}{\Gamma(-\alpha)}\intop_{B_{r}}d_{p}y\frac{\varphi(y,t)-\varphi(x,t)}{\left|y-x\right|_{p}^{\alpha+1}}.\label{Vlad}
\end{equation}

\begin{theorem}
Let $\varphi\left(x\right)\subset L^{1}\left(B_{r},d_{p}x\right)$. Then the function
\begin{equation}
\tilde{\varphi}\left(k\right)=p^{-r/2}\intop_{B_{r}}\chi\left(kx\right)\varphi\left(x\right)d_{p}x\label{phi(k)}
\end{equation}
lies in $W_{-r}^{0}$ {\rm (}that is, it is a function of $k\in\mathbb{Q}_{p}/B_{-r})$. Moreover,
\begin{equation}
\varphi\left(x\right)=p^{-r/2}\sum_{k\in\mathbb{Q}_{p}/B_{-r}}\chi\left(-kx\right)\tilde{\varphi}\left(k\right).\label{phi(x)}
\end{equation}
\label{th2}
\end{theorem}

\begin{proof}
Let $k'\in\mathbb{Q}_{p}$, $|k'|_{p}\le p^{-r}$. Then $\chi\left(k'x\right)=1$ with $|x|_{p}\le p^{r}$, and besides
\begin{equation}
\tilde{\varphi}\left(k+k'\right)=p^{-r/2}\intop_{B_{r}}\chi\left(kx\right)\chi\left(k'x\right)\varphi\left(x\right)d_{p}x=\tilde{\varphi}\left(k\right).\label{lc-1}
\end{equation}
Thus, $\tilde{\varphi}\left(k\right)$ is also locally constant with the exponent of local constancy  $l=-r$. Hence, we may assume that
 $k\in\mathbb{Q}_{p}/B_{-r}$.

For any $s>-r$ we consider the function
\[
\varphi_{s}\left(x\right)=p^{-r/2}\sum_{k\in B_{s}/B_{-r}}\chi\left(-kx\right)\tilde{\varphi}\left(k\right)=p^{-r/2}\sum_{k\in B_{s}/B_{-r}}\chi\left(-kx\right)\left[p^{-r/2}\intop_{B_{r}}\chi\left(ky\right)\varphi\left(y\right)d_{p}y\right].
\]
Since the integral converges uniformly in $k$, we have
\[
\varphi_{s}\left(x\right)=p^{-r}\intop_{B_{r}}\varphi\left(y\right)\left[\sum_{k\in B_{s}/B_{-r}}\chi\left(k\left(y-x\right)\right)\right]d_{p}y=
\]
\[
=\intop_{B_{r}}\varphi\left(y\right)\left[\intop_{B_{s}}\chi\left(k\left(y-x\right)\right)d_{p}k\right]d_{p}y=\intop_{B_{r}}\varphi\left(y\right)\left[p^{s}\Omega\left(\left|x-y\right|_{p}p^{s}\right)\right]d_{p}y.
\]
Taking into account that
\[
\intop_{B_{r}}\left[p^{s}\Omega\left(\left|x-y\right|_{p}p^{s}\right)\right]d_{p}y=1,
\]
this establishes
\[
\varphi_{s}\left(x\right)-\varphi\left(x\right)=\intop_{B_{r}}\left(\varphi\left(y\right)-\varphi\left(x\right)\right)\left[p^{s}\Omega\left(\left|x-y\right|_{p}p^{s}\right)\right]d_{p}y=
\]
\[
=\intop_{|x-y|_{p}\leq p^{-s}}\left(\varphi\left(y\right)-\varphi\left(x\right)\right)d_{p}y\rightarrow0
\]
as  $s\rightarrow\infty$. It follows that $\lim_{s\rightarrow\infty}\varphi_{s}\left(x\right)=\varphi\left(x\right)$
and so we arrive to~(\ref{phi(x)}), completing the proof of Theorem~\ref{th2}.
\end{proof}

From Theorem 2 it follows that the functions  $\chi\left(kx\right)$ form a~complete
system of functions in $L^{1}\left(B_{r},d_{p}x\right)$.
These functions are mutually orthogonal in $L^{2}\left(B_{r},d_{p}x\right)$. Moreover, since
\[
\intop_{B_{r}}d_{p}x\left[p^{-r/2}\chi\left(-k'x\right)\right]\left[p^{-r/2}\chi\left(kx\right)\right]=p^{-r}\intop_{B_{r}}d_{p}x\chi\left(\left(k-k'\right)x\right)=\delta_{k,k'}
\]
we have the following result.

\begin{corollary}
The family of functions $p^{-r/2}\chi\left(kx\right)$,
$k=\cdots+a_{i-1}p^{i-1}+a_{i}p^{i}+a_{i+1}p^{i+1}+\cdots+a_{r-1}p^{r-1}\in\mathbb{Q}_{p}/B_{-r}$,
$a_{-}=0,1,\ldots.p-1$, forms a~countable orthonormal basis for $L^{2}\left(B_{r},d_{p}x\right)$.
\label{col1}
\end{corollary}

\begin{theorem}
Let $f\left(x\right),\:g\left(x\right)\in L^{2}\left(B_{r},d_{p}x\right)$ and let
$$
\tilde{f}\left(k\right)=p^{-r/2}\intop_{B_{r}}\chi\left(kx\right)f\left(x\right)d_{p}x,\qquad
\tilde{g}\left(k\right)=p^{-r/2}\intop_{B_{r}}\chi\left(kx\right)g\left(x\right)d_{p}x.$$
Then
\begin{equation}
f\left(x\right)g(x)=p^{-r/2}\sum_{\zeta\in\mathbb{Q}_{p}/B_{-r}}\chi\left(-\zeta x\right)\widetilde{\left[fg\right]}(\zeta),\label{conv}
\end{equation}
where
\begin{equation}
\widetilde{\left[fg\right]}(\zeta)=p^{-r/2}\sum_{k\in\mathbb{Q}_{p}/B_{-r}}\tilde{f}\left(k\right)\tilde{g}\left(\zeta-k\right)\label{conv1}
\end{equation}
and
\begin{equation}
\intop_{B_{r}}f\left(y\right)g\left(x-y\right)d_{p}y=\sum_{k\in\mathbb{Q}_{p}/B_{-r}}\chi\left(-kx\right)\tilde{f}\left(k\right)\tilde{g}\left(k\right).\label{conv2}
\end{equation}
\label{th3}
\end{theorem}

\begin{proof}
For any $s>-r$ we consider the functions  $f_{s}\left(x\right)=p^{-r/2}\sum_{k\in B_{s}/B_{-r}}\chi\left(-kx\right)\tilde{f}\left(k\right)$
and $g_{s}\left(x\right)=p^{-r/2}\sum_{k\in B_{s}/B_{-r}}\chi\left(-kx\right)\tilde{g}\left(k\right)$. We have
\[
f\left(x\right)g(x)=\lim_{s\rightarrow\infty}f_{s}\left(x\right)g_{s}\left(x\right).
\]
Each of the sums in $f_{s}\left(x\right)$ and $g_{s}\left(x\right)$ contains a~finite number of terms, and hence, writing
\[
f\left(x\right)g\left(x\right)=p^{-r}\lim_{s\rightarrow\infty}\sum_{\xi\in B_{s}/B_{-r}}\sum_{k\in B_{s}/B_{-r}}\chi\left(-(k+\xi)x\right)\tilde{f}\left(k\right)\tilde{g}\left(\xi\right),
\]
one may change variables to  $k+\xi=\zeta\in B_{s}/B_{-r}$. Consequently, making $s\rightarrow\infty$, we see that
\[
f\left(x\right)g(x)=p^{-r/2}\sum_{\zeta\in Q_{p}/B_{-r}}\chi\left(-\zeta x\right)p^{-r/2}\sum_{k\in\mathbb{Q}_{p}/B_{-r}}\tilde{f}\left(k\right)\tilde{g}\left(\zeta-k\right),
\]
which completes the proof of (\ref{conv1}).

In order to prove (\ref{conv2}) we observe that
\[
\intop_{B_{r}}f\left(y\right)g\left(x-y\right)d_{p}x=\lim_{s\rightarrow\infty}\intop_{B_{r}}f_{s}\left(y\right)g_{s}\left(x-y\right)d_{p}x=
\]
\[
=p^{-r}\lim_{s\rightarrow\infty}\intop_{B_{r}}\left[\sum_{k\in B_{s}/B_{-r}}\chi\left(-ky\right)\tilde{f}\left(k\right)\right]\left[\sum_{k'\in B_{s}/B_{-r}}\chi\left(-k'\left(x-y\right)\right)\tilde{g}\left(k'\right)\right]d_{p}x=
\]
\[
=p^{-r}\lim_{s\rightarrow\infty}\intop_{B_{r}}\left[\sum_{k\in B_{s}/B_{-r}}\sum_{k'\in B_{s}/B_{-r}}\chi\left(\left(k-k'\right)y\right)\tilde{f}\left(k\right)\right]\left[\chi\left(-k'x\right)\tilde{g}\left(k'\right)\right]d_{p}x.
\]
Since the sums are finite, we may write
\[
\intop_{B_{r}}f\left(y\right)g\left(x-y\right)d_{p}x=
\]
\[
=p^{-r}\lim_{s\rightarrow\infty}\sum_{k\in B_{s}/B_{-r}}\sum_{k'\in B_{s}/B_{-r}}\left[\intop_{B_{r}}\chi\left(\left(k-k'\right)y\right)d_{p}x\right]\tilde{f}\left(k\right)\tilde{g}\left(k'\right)\chi\left(-k'x\right).
\]
Taking into account that $\intop_{B_{r}}\chi\left(\left(k-k'\right)y\right)d_{p}x=p^{r}\delta_{k,k'}$ and making $s\rightarrow\infty$, we obtain
\[
\intop_{B_{r}}f\left(y\right)g\left(x-y\right)d_{p}x=
\]
\[
=\sum_{k\in\mathbb{Q}_{p}/B_{-r}}\tilde{f}\left(k\right)\tilde{g}\left(k\right)\chi\left(-kx\right),
\]
concluding the proof of the theorem.
\end{proof}

\begin{lemma}
The functions $p^{-r/2}\chi\left(kx\right)$, $k\in\mathbb{Q}_{p}/B_{-r}$ are eigenfunctions of operator (\ref{Vlad}) with eigenvalues
\[
-|k|_{p}^{\alpha}+(1-p^{-1})\frac{p^{-\alpha r}}{1-p^{-\alpha-1}}\left(1-\delta_{k,0}\right).
\]
\label{lem1}
\end{lemma}

\begin{proof}
Consider the action of  $D^{\alpha}(B_{r})$ on $\chi(kx)$ with $k\in\mathbb{Q}_{p}/B_{-r}$ and $x\in B_{r}$. If
$k=0$, then  $D^{\alpha}(B_{r})\chi(kx)=0$. Next, if $|k|_{p}>p^{-r}$, then
\[
D^{\alpha}(B_{r})\chi(kx)=-\frac{1}{\Gamma(-\alpha)}\intop_{B_{r}}d_{p}y\frac{\Omega\left(|y|_{p}p^{-r}\right)\chi(ky)-\Omega\left(|x|_{p}p^{-r}\right)\chi(kx)}{\left\vert y-x\right\vert _{p}^{\alpha+1}}=
\]
\[
=-\frac{1}{\Gamma(-\alpha)}\intop_{B_{r}}d_{p}y\frac{\chi(ky)-\chi(kx)}{\left\vert y-x\right\vert _{p}^{\alpha+1}}=-\frac{\chi(kx)}{\Gamma(-\alpha)}\int_{B_{r}}d_{p}y\frac{\chi(ky)-1}{\left\vert y\right\vert _{p}^{\alpha+1}}.
\]
Evaluating the integral, we get
\[
\intop_{B_{r}}d_{p}y\frac{\chi(ky)-1}{\left\vert y\right\vert _{p}^{\alpha+1}}=\left(|k|_{p}^{\alpha}\Gamma(-\alpha)+(1-p^{-1})\frac{p^{-\alpha(r+1)}}{1-p^{-\alpha}}\right).
\]
As a result,
\[
D^{\alpha}(B_{r})\chi(kx)=-\left(|k|_{p}^{\alpha}-(1-p^{-1})\frac{p^{-\alpha r}}{1-p^{-\alpha-1}}\right)\chi(kx),
\]
This completes the proof of Lemma~\ref{lem1}.
\end{proof}

\section{A real basis for $L^{2}\left(B_{r}\right)$}

From Corollary \ref{col1} and Lemma \ref{lem1} we have the following result.

\begin{corollary}
Let $\phi_{\gamma}\left(k\right)$ be some family of functions from $L^{1}\left(\mathbb{Q}_{p}/B_{-r}\right)$,
$\gamma\leq r$, and let $\phi_{\gamma}\left(k\right)=\phi_{\gamma}^{\ast}\left(-k\right)$.
Then the functions
\[
f_{\gamma}\left(x\right)=\sum_{k\in\mathbb{Q}_{p}/B_{-r}}
\delta\left(|k|_{p}-p^{-\gamma+1}\right)\phi_{\gamma}\left(k\right)\chi(kx),\;\gamma\leq r,
\]
are orthogonal real eigenfunctions of operator (\ref{Vlad}) with eigenvalues
\begin{equation}
\lambda_{\gamma}=-p^{-\alpha\left(\gamma-1\right)}+(1-p^{-1})\frac{p^{-\alpha r}}{1-p^{-\alpha-1}}.
\label{EV}
\end{equation}
\label{col2}
\end{corollary}
In particular, taking
\[
\phi_{\gamma}\left(k\right)=p^{-r+\gamma-1},
\]
and employing the relation
\[
\sum_{k\in\mathbb{Q}_{p}/B_{-r}}\delta\left(|k|_{p}-p^{-\gamma+1}\right)\chi\left(-kx\right)=p^{r-\gamma+1}(1-p^{-1})\Omega\left(|x|_{p}p^{-\gamma+1}\right)-p^{r-\gamma}\delta\left(|x|_{p}-p^{\gamma}\right),
\]
we obtain
\begin{equation}
f_{\gamma}\left(x\right)=(1-p^{-1})\Omega\left(|x|_{p}p^{-\gamma+1}\right)-p^{-1}\delta\left(|x|_{p}-p^{\gamma}\right)\label{f_i}
\end{equation}
or
\begin{equation}
f_{\gamma}\left(x\right)=\Omega\left(|x|_{p}p^{-\gamma+1}\right)-p^{-1}\Omega\left(|x|_{p}p^{-\gamma}\right).\label{f_i-1}
\end{equation}
From the relations
\begin{equation}
\intop_{B_{r}}\left[\Omega\left(|x|_{p}p^{-\gamma+1}\right)-p^{-1}\Omega\left(|x|_{p}p^{-\gamma}\right)\right]^{2}d_{p}x=p^{\gamma-1}\left(1-p^{-1}\right),\label{f^2}
\end{equation}
\begin{equation}
\intop_{B_{r}}\left[\Omega\left(|x|_{p}p^{-\gamma+1}\right)-p^{-1}\Omega\left(|x|_{p}p^{-\gamma}\right)\right]\left[\Omega\left(|x|_{p}p^{-\delta+1}\right)-p^{-1}\Omega\left(|x|_{p}p^{-\delta}\right)\right]d_{p}x=p^{\gamma-1}\left(1-p^{-1}\right)\delta_{\gamma\delta}\label{f_f}
\end{equation}
it follows that the family of functions (\ref{f_i-1}) is the orthogonal family of eigenfunctions of operator (\ref{Vlad}) with eigenvalues~(\ref{EV}).

Since the Vladimirov operator  $D^{\alpha}\left(B_{r}\right)$ is invariant with respect to translations of coordinates and since any ball  $B_{r}$ of radius~$\gamma$
can be obtained from the ball  $B_{\gamma}$ by the translation $x\rightarrow x-n$,
where $n\in B_{r}/B_{\gamma}$, it follows that the functions
\begin{equation}
f_{\gamma,n}\left(x\right)=f_{\gamma}\left(x-n\right)=\left[\Omega\left(|x-n|_{p}p^{-\gamma+1}\right)-p^{-1}\Omega\left(|x-n|_{p}p^{-\gamma}\right)\right],\;\gamma\leq r,\;n\in B_{r}/B_{\gamma},\label{f_i_n}
\end{equation}
as obtained from functions (\ref{f_i-1}) by the translation $x\rightarrow x-n$,
are eigenfunctions of operator (\ref{Vlad}) with eigenvalues~(\ref{EV}). Next, the translation $x\rightarrow x-ap^{-\gamma}$,
$a=0,1,\ldots,p-1$, leaves the function $\Omega\left(|x-n|_{p}p^{-\gamma}\right)$ invariant, but changes the function $\Omega\left(|x-n|_{p}p^{-\gamma+1}\right)$.
Hence, the functions
\[
f_{\gamma,n,a}\left(x\right)=f_{\gamma,n}\left(x-ap^{-\gamma}\right)=f_{\gamma}\left(x-n-ap^{-\gamma}\right)=
\]
\begin{equation}
=\left[\Omega\left(|x-n{}^{-\gamma}-ap^{-\gamma}|_{p}p^{-\gamma+1}\right)-p^{-1}\Omega\left(|x-n|_{p}p^{-\gamma}\right)\right],\label{f_i_n_a}
\end{equation}
where $\gamma\leq r$ and $n\in B_{r}/B_{\gamma}$, are also eigenfunctions of operator~(\ref{Vlad}) with eigenvalues~(\ref{EV}).

\begin{theorem}
The functions (\ref{f_i_n_a}) and the constant function
\begin{equation}
f_{r}\left(x\right)=p^{-r}\label{fx}
\end{equation}
form a~complete system of vectors in the space $L^{2}\left(B_{r}\right)$.
\label{th4}
\end{theorem}

\begin{proof}
Since the set of characteristic functions of all balls forms a basis for $L^{2}\left(B_{r}\right)$, it suffices to prove that the characteristic function $\varOmega\left(|y-n'|_{p}p^{-\gamma'}\right)$
of any ball $B_{\gamma'}\left(n'\right)$, $\gamma'\leq r,\:n'\in B_{r}/B_{\gamma'}$ from $B_{r}$ can be represented as a~linear combination of functions
from~(\ref{f_i_n_a}):
\begin{equation}
C_{r}f_{r}\left(x\right)+\sum_{\gamma=-\infty}^{r}\sum_{n\in B_{r}/B_{\gamma}}\sum_{b=1}^{p-1}C_{\gamma,n,b}f_{\gamma,n,b}\left(x\right)=\varOmega\left(|x-n'|_{p}p^{-\gamma'}\right).\label{decom}
\end{equation}
Next, since the set of functions $f_{\gamma,n,b}\left(x\right)$
is invariant under transformations  $x\rightarrow p^{-\gamma'}\left(x-n\right)'$,
it suffices to check (\ref{decom}) with fixed $\gamma'$ and~$n'$ (for example, with $\gamma'=0$ and $n'=0$; that is, for $\varOmega\left(|x|\right)$).
It is easily seen that
\begin{equation}
\Omega\left(|x|_{p}\right)=f_{r}\left(x\right)+\sum_{\gamma=1}^{r}p^{-\gamma+1}f_{\gamma,0,0}\left(x\right),\label{omega}
\end{equation}
which concludes the proof of Theorem~\ref{th4}.
\end{proof}

\begin{remark}
The complete system of functions $\left\{ f_{r},\left\{ f_{\gamma,n,a}\left(x\right)\right\} \right\} $
is overcomplete in $L^{2}\left(B_{r}\right)$, because in a~family  of~$p$ functions $f_{\gamma,n,a}\left(x\right)$, where $a=0,1,\ldots,p-1$
and  $i$ and $n$~are fixed, only $p-1$ functions are linearly independent. This follows from the relation
\[
\sum_{a=0}^{p-1}f_{\gamma,n,a}\left(x\right)=0.
\]
\label{rem1}
\end{remark}

The next theorem gives a construction of an orthonormal basis for $L^{2}\left(B_{r}\right)$.

\begin{theorem}
The functions
\begin{equation}
\varphi_{\gamma,n,b}\left(x\right)=\dfrac{p^{-\left(\gamma-1\right)/2}}{k}\left(f_{\gamma,n,0}\left(x\right)+kf_{\gamma,n,b}\left(x\right)\right),\:\gamma\leq r,\:b=1,\ldots,p-1\label{bas1}
\end{equation}
and the function
\begin{equation}
\varphi_{r}\left(x\right)=p^{-r/2}\label{bas2},
\end{equation}
where $k=-1\pm\sqrt{p}$, form an orthonormal basis for $L^{2}\left(B_{r}\right)$.
\label{th5}
\end{theorem}

\begin{proof}
Since the functions $f_{\gamma,n,a}\left(x\right)$ and $f_{r}=p^{-r}$ are orthogonal with different $i$ and~$n$, it follows that
the functions  $\varphi_{\gamma,n,b}\left(x\right)$ and $\varphi_{r}\left(x\right)$ are also orthogonal with different $i$ and~$n$. We have
\[
\intop_{B_{r}}d_{p}x\varphi_{\gamma,n,b}\left(x\right)\varphi_{\gamma,n,b'}\left(x\right)=
\]
\[
=\dfrac{p^{-\gamma+1}}{k^{2}}\intop_{B_{r}}d_{p}x\left[f_{\gamma,n,0}\left(x\right)+kf_{\gamma,n,b}\left(x\right)\right]\left[f_{\gamma,n,0}\left(x\right)+kf_{\gamma,n,b'}\left(x\right)\right]=
\]
\[
=\dfrac{p^{-\gamma+1}}{k^{2}}\intop_{B_{r}}d_{p}x\left[f_{\gamma,n,0}\left(x\right)f_{\gamma,n,0}\left(x\right)+kf_{\gamma,n,b}\left(x\right)f_{\gamma,n,0}\left(x\right)+kf_{\gamma,n,b'}\left(x\right)f_{\gamma,n,0}\left(x\right)+k^{2}f_{\gamma,n,b}\left(x\right)f_{\gamma,n,b'}\left(x\right)\right].
\]
Taking into account that
\[
\intop d_{p}xf_{\gamma,n,a}\left(x\right)f_{\gamma,n,b}\left(x\right)=
\]
\[
=\intop_{B_{r}}d_{p}x\left[\Omega\left(|x-n-ap^{-\gamma}|_{p}p^{-\gamma+1}\right)-p^{-1}\Omega\left(|x-n|_{p}p^{-\gamma}\right)\right]\times
\]
\[
\times\left[\Omega\left(|x-n-bp^{-\gamma}|_{p}p^{-\gamma+1}\right)-p^{-1}\Omega\left(|x-n|_{p}p^{-\gamma}\right)\right]=
\]
\[
=\left(p^{\gamma-1}\delta_{ab}-2p^{-1}p^{\gamma-1}+p^{-2}p^{\gamma}\right)=p^{\gamma-1}\left(\delta_{ab}-p^{-1}\right)
\]
we get
\[
\intop_{B_{r}}d_{p}x\varphi_{\gamma,n,b}\left(x\right)\varphi_{\gamma,n,b'}\left(x\right)=
\]
\[
=\dfrac{p^{-\gamma+1}}{k^{2}}p^{\gamma-1}\left(1-p^{-1}-2p^{-1}k+\left(\delta_{b,b'}-p^{-1}\right)k^{2}\right)=
\]
\[
=\dfrac{1}{k^{2}}\left(p^{-1}\left(p-1-2k-k^{2}\right)+\delta_{b,b'}k^{2}\right)=\delta_{b,b'},
\]
because $k=-1\pm\sqrt{p}$ is a~root of the equation $p-1-2k-k^{2}=0$.
This completes the proof of Theorem~\ref{th5}.
\end{proof}

\begin{remark}
We note the following formulas for the inverse transition from $\varphi_{\gamma,n,b}\left(x\right)$ to $f_{\gamma,n,a}\left(x\right)$:
\[
f_{\gamma,n,0}\left(x\right)=\dfrac{kp^{\left(\gamma-1\right)/2}}{p-k-1}\sum_{b=1}^{p-1}\varphi_{\gamma,n,b}\left(x\right),
\]
\[
f_{\gamma,n,b}\left(x\right)=p^{\left(\gamma-1\right)/2}\left(\varphi_{\gamma,n,b}\left(x\right)-\dfrac{1}{p-k-1}\sum_{b'=1}^{p-1}\varphi_{\gamma,n,b'}\left(x\right)\right),\:b=1,\ldots,p-1.
\]
\label{rem2}
\end{remark}

To demonstrate one of the advantages of the above basis~(\ref{bas1}), (\ref{bas2}), we consider the Cauchy problem for the Vladimirov equation
\begin{equation}
\dfrac{d}{dt}f\left(x,t\right)=D^{\alpha}\left(B_{r}\right)f\left(x\right)\label{Vl}
\end{equation}
with the initial condition
\[
f\left(x,0\right)=\Omega\left(|x|_{p}\right).
\]
Formally, the solution of equation (\ref{Vl}) can be written as
\[
f\left(x,t\right)=\exp\left(D^{\alpha}\left(B_{r}\right)t\right)f\left(x,0\right).
\]
Using expansion (\ref{omega})  we easily find the solution
\[
f\left(x,t\right)=\exp\left(D^{\alpha}\left(B_{r}\right)t\right)\left(f_{r}\left(x\right)+\sum_{\gamma=1}^{r}p^{-\gamma+1}f_{\gamma,0,0}\left(x\right)\right)=
\]
\begin{equation}
=\dfrac{1}{p^{r}}+\sum_{i=0}^{r-1}p^{-i}\left(\Omega\left(|x|_{p}p^{-i}\right)-p^{-1}\Omega\left(|x|_{p}p^{-i-1}\right)\right)\exp\left(-\left(p^{-i\alpha}-\left(1-p^{-1}\right)\dfrac{p^{-r\alpha}}{1-p^{-\alpha-1}}\right)t\right).\label{solution}
\end{equation}

\section{Real basis for $L^{2}\left(\mathbb{Q}_{p}\right)$ and its relation with the wavelet basis}

The following result holds.

\begin{theorem}
The functions
\begin{equation}
f_{\gamma,n,a}\left(x\right)=\left[\Omega\left(|x-np^{-\gamma}-ap^{-\gamma}|_{p}p^{-\gamma+1}
\right)-p^{-1}\Omega\left(|x-np^{-\gamma}|_{p}p^{-\gamma}\right)\right],\label{f_i_n_a_Q_p}
\end{equation}
where $\gamma\in\mathbb{Z}$, $n\in\mathbb{Q}_{p}/Z_{p}$ and $a=0,1,\ldots,p-1$,
are eigenfunctions of the Vladimirov operator (\ref{Vlad_Q_p}) on~$\mathbb{Q}_{p}$  with eigenvalues
\begin{equation}
\lambda_{\gamma}=-p^{-\alpha\left(\gamma-1\right)}\label{EV_Q_p}
\end{equation}
and form a~complete system of vectors in the space $L^{2}\left(\mathbb{Q}_{p}\right)$.
\label{th6}
\end{theorem}

\begin{proof}
Since $\mathbb{Q}_{p}$ may be looked upon as the limit case of $B_{r}$ as $r\rightarrow\infty$, repeating the argument of the
previous section it is easily seen that functions  (\ref{f_i_n_a_Q_p})
are eigenfunctions of operator~(\ref{Vlad_Q_p}) with eigenvalues~(\ref{EV_Q_p}).

Next, since the \_\_[[[set of characteristic functions]]] \_\_ of all balls forms a basis for $L^{2}\left(\mathbb{Q}_{p}\right)$,
it suffices to show that the characteristic function $\varOmega\left(|x-n'p^{-\gamma'}|_{p}p^{-\gamma'}\right)$
of any ball  $B_{\gamma'}\left(n'p^{-\gamma'}\right)$, $\gamma'\in\mathbb{Z},\:n'\in\mathbb{Q}_{p}/Z_{p}$
from $\mathbb{Q}_{p}$ can be expanded into functions~(\ref{f_i_n_a}).
Since the set of functions $f_{\gamma,n,b}\left(x\right)$ is invariant under transformations $x\rightarrow p^{-\gamma'}x-n'$,
it will suffice to show that such an expansion holds for any fixed
$\gamma'$ and~$n'$; for example, for  $\gamma'=0$, $n'=0$; that is,
\begin{equation}
\sum_{\gamma=-\infty}^{\infty}\sum_{n\in\mathbb{Q}_{p}/Z_{p}}\sum_{b=1}^{p-1}C_{\gamma,n,b}f_{\gamma,n,b}\left(x\right)=\varOmega\left(|x|\right).\label{decomp}
\end{equation}
Consider the function $g\left(x\right)$ represented by the series
\begin{equation}
g\left(x\right)=\sum_{\gamma=1}^{\infty}p^{-\gamma+1}f_{\gamma,0,0}\left(x\right).\label{g(x)},
\end{equation}
which is uniformly convergent in~$x$.  Next, consider the squared norm
\[
\left\Vert g\left(x\right)-\varOmega\left(|x|\right)\right\Vert ^{2}=\intop_{\mathbb{Q}_{p}}d_{p}x\left(g\left(x\right)-\varOmega\left(|x|\right)\right)^{2}.
\]
Since the series in (\ref{g(x)}) is uniformly convergent and using  (\ref{f^2}) and (\ref{f_f}), we find that
\[
\left\Vert g\left(x\right)-\varOmega\left(|x|\right)\right\Vert ^{2}=
\]
\[
=\sum_{\gamma=1}^{\infty}p^{-\gamma+1}\sum_{\delta=1}^{\infty}p^{-\delta+1}\intop_{\mathbb{Q}_{p}}d_{p}xf_{\gamma,0,0}\left(x\right)f_{\delta,0,0}\left(x\right)-2\sum_{\gamma=1}^{\infty}p^{-\gamma+1}\intop_{\mathbb{Q}_{p}}d_{p}xf_{\gamma,0,0}\left(x\right)\varOmega\left(|x|\right)+\intop_{\mathbb{Q}_{p}}d_{p}x\varOmega\left(|x|\right)=
\]
\[
=\sum_{\gamma=1}^{\infty}p^{-\gamma+1}\sum_{\delta=1}^{\infty}p^{-\delta+1}\left(p^{\gamma-1}\left(1-p^{-1}\right)\delta_{\gamma\delta}\right)-2\sum_{\gamma=1}^{\infty}p^{-\gamma+1}\left(1-p^{-1}\right)+1=
\]
\[
=\left(1-p^{-1}\right)\sum_{\gamma=1}^{\infty}p^{-\gamma+1}-2\left(1-p^{-1}\right)\sum_{\gamma=1}^{\infty}p^{-\gamma+1}+1=
\]
\[
=-\left(1-p^{-1}\right)\sum_{\gamma=0}^{\infty}p^{-\gamma}+1=0.
\]
The norm of $g\left(x\right)-\varOmega\left(|x|\right)$ is zero, and hence
\[
g\left(x\right)=\varOmega\left(|x|\right),
\]
which shows that
\begin{equation}
\varOmega\left(|x|\right)=\sum_{\gamma=1}^{\infty}p^{-\gamma+1}f_{\gamma,0,0}\left(x\right)\label{omega_dec}.
\end{equation}
Consequently, (\ref{decomp})  holds. This completes the proof of the theorem.
\end{proof}

\begin{theorem}
The functions
\begin{equation}
\varphi_{\gamma,n,b}\left(x\right)=\dfrac{p^{-\left(\gamma-1\right)/2}}{k}
\left(f_{\gamma,n,0}\left(x\right)+kf_{\gamma,n,b}\left(x\right)\right),\:
\gamma\in\mathbb{Z},\:n\in\mathbb{Q}_{p}/Z_{p}, \ \ b=1,\ldots,p-1\label{bas}
\end{equation}
where $k=-1\pm\sqrt{p}$, form an orthonormal basis for $L^{2}\left(\mathbb{Q}_{p}\right)$.
\label{th7}
\end{theorem}

The proof of Theorem~\ref{th7} is similar to that of Theorems 5 and~6.

\smallskip

Let us find a relation between the basis introduced above and the basis of $p$-adic wavelets in $L^{2}\left(\mathbb{Q}_{p}\right)$, which was proposed in \cite{Koz,Koz1,KK}.
The basis of $p$-adic wavelets is formed by the functions
\[
\psi_{\gamma,n,j}\left(x\right)=p^{-\gamma/2}\chi\left(p^{\gamma-1}j\left(x-p^{-\gamma}n\right)\right)\Omega\left(\left|p^{\gamma}x-n\right|_{p}\right),
\]
\begin{equation}
\gamma\in\mathbb{Z},\;n\in\mathbb{Q}_{p}/\mathbb{Z}_{p},\;j=1,\ldots,p-1.\label{V-B}
\end{equation}
Summing (\ref{V-B}) in $j$, we find that
\[
\sum_{j=1}^{p-1}\psi_{\gamma,n,j}\left(x\right)=p^{-\gamma/2}\Omega\left(\left|p^{\gamma}x-n\right|_{p}\right)\sum_{j=1}^{p-1}\chi\left(p^{\gamma-1}j\left(x-p^{-\gamma}n\right)\right)=
\]
\[
=p^{-\gamma/2}\left(\varOmega\left(\left|x-p^{-\gamma}n\right|p^{-\gamma+1}\right)-\delta\left(\left|x-p^{-\gamma}n\right|-p^{\gamma}\right)\right)\sum_{j=1}^{p-1}\chi\left(p^{\gamma-1}j\left(x-p^{-\gamma}n\right)\right)=
\]
\[
=p^{-\gamma/2+1}\left((1-p^{-1})\varOmega\left(\left|x-p^{-\gamma}n\right|p^{-\gamma+1}\right)-p^{-1}\delta\left(\left|x-p^{-\gamma}n\right|-p^{\gamma}\right)\right)=
\]
\[
=p^{-\gamma/2+1}\left(\varOmega\left(\left|x-p^{-\gamma}n\right|p^{-\gamma+1}\right)-p^{-1}\Omega\left(\left|x-p^{-\gamma}n\right|p^{-\gamma}\right)\right)=p^{-\gamma/2+1}f_{\gamma,n,0}\left(x\right).
\]
As a result, we may express $f_{\gamma,n,0}\left(x\right)$ as
\[
f_{\gamma,n,0}\left(x\right)=p^{\gamma/2-1}\sum_{j=1}^{p-1}\psi_{\gamma,n,j}\left(x\right)
\]
and find the functions $f_{\gamma,n,a}\left(x\right)$ and $\varphi_{\gamma,n,b}\left(x\right)$
with $\gamma\leq r,\:a=0,1,\ldots,p-1,\:b=1,\ldots,p-1$,
\[
f_{\gamma,n,a}\left(x\right)=p^{\gamma/2-1}\sum_{j=1}^{p-1}\psi_{\gamma,n,j}\left(x-ap^{-\gamma}\right),
\]
\[
\varphi_{\gamma,n,b}\left(x\right)=\dfrac{p^{-1/2}}{k}\sum_{j=1}^{p-1}
\left(\psi_{\gamma,n,j}\left(x\right)+k\psi_{\gamma,n,j}\left(x-bp^{-\gamma}\right)\right).
\]
We have
\[
\psi_{\gamma,n,j}\left(x-bp^{-\gamma}\right)=\psi_{\gamma,n,j}\left(x\right)\chi\left(p^{-1}jb\right),
\]
and hence,
\begin{equation}
f_{\gamma,n,a}\left(x\right)=p^{\gamma/2-1}\sum_{j=1}^{p-1}\chi\left(p^{-1}ja\right)\psi_{\gamma,n,j}
\left(x\right),\label{f-psi}
\end{equation}
\begin{equation}
\varphi_{\gamma,n,b}\left(x\right)=\dfrac{p^{-1/2}}{k}\sum_{j=1}^{p-1}\left(1+k\chi\left(p^{-1}jb\right)\right)\psi_{\gamma,n,j}\left(x\right).\label{phi-psi}
\end{equation}
Let us make the inverse transform from  $\varphi_{\gamma,n,b}\left(x\right)$ to $\psi_{\gamma,n,j}\left(x\right)$. We have
\[
\sum_{a=0}^{p-1}\chi\left(p^{-1}ja\right)\chi\left(-p^{-1}j'a\right)=\sum_{a=0}^{p-1}\chi\left(p^{-1}a\left(j-j'\right)\right)=p\delta_{j,j'},
\]
and hence,
\[
\sum_{a=0}^{p-1}\chi\left(-p^{-1}j'a\right)f_{\gamma,n,a}\left(x\right)=p^{\gamma/2-1}\sum_{a=0}^{p-1}\sum_{j=1}^{p-1}\chi\left(-p^{-1}j'a\right)\chi\left(p^{-1}ja\right)\psi_{\gamma,n,j}\left(x\right)=p^{\gamma/2}\psi_{\gamma,n,j'}\left(x\right).
\]
Therefore,
\begin{equation}
\psi_{\gamma,n,j}\left(x\right)=p^{-\gamma/2}\sum_{a=0}^{p-1}\chi\left(-p^{-1}j'a\right)f_{\gamma,n,a}\left(x\right).\label{psi-f}
\end{equation}
Expressing in (\ref{psi-f}) $f_{\gamma,n,a}\left(x\right)$ it terms of $\varphi_{\gamma,n,b}\left(x\right)$, this gives
\[
\psi_{\gamma,n,j}\left(x\right)=p^{-\gamma/2}f_{\gamma,n,0}\left(x\right)+p^{-\gamma/2}\sum_{b=1}^{p-1}\chi\left(-p^{-1}jb\right)f_{\gamma,n,a}\left(x\right)=
\]
\[
=\dfrac{kp^{-1/2}}{p-k-1}\sum_{b=1}^{p-1}\varphi_{\gamma,n,b}\left(x\right)+p^{-1/2}\sum_{b=1}^{p-1}\chi\left(-p^{-1}jb\right)\left(\varphi_{\gamma,n,b}\left(x\right)-\dfrac{1}{p-k-1}\sum_{b'=1}^{p-1}\varphi_{\gamma,n,b'}\left(x\right)\right)=
\]
\[
=\dfrac{kp^{-1/2}}{p-k-1}\sum_{b=1}^{p-1}\varphi_{\gamma,n,b}\left(x\right)+p^{-1/2}\left(\sum_{b=1}^{p-1}\chi\left(-p^{-1}jb\right)\varphi_{\gamma,n,b}\left(x\right)+\dfrac{1}{p-k-1}\sum_{b=1}^{p-1}\varphi_{\gamma,n,b}\left(x\right)\right)
\]
Finally, we have
\begin{equation}
\psi_{\gamma,n,j}\left(x\right)=p^{-1/2}\sum_{b=1}^{p-1}\left(\dfrac{k+1}{p-k-1}+\chi\left(-p^{-1}jb\right)\right)\varphi_{\gamma,n,b}\left(x\right).\label{psi-fi}
\end{equation}
Formulas (\ref{f-psi})--(\ref{psi-fi}) establish a relation between the
basis introduced above and the basis of wavelets for $L^{2}\left(\mathbb{Q}_{p}\right)$.

In \cite{Koz,Koz1} it was shown that the wavelet basis can be used to diagonalize not only
the pseudo-differential Vladimirov operator but also the class of pseudo-differential operators of the form
\begin{equation}
Tf\left(x\right)=\intop_{\mathbb{Q}_{p}}T\left(x,y\right)\left(f(y)-f(x)\right)dy,\label{T}
\end{equation}
where the kernel $T\left(x,y\right)$ satisfies the conditions:
\begin{enumerate}
\item $T\left(x,y\right)=T\left(y,x\right),$
\item $T\left(x,y\right)=\operatorname{const}$ on the set of points satisfying the equation
$|x-y|_{p}=\operatorname{const}$.
\end{enumerate}
Such a kernel can be written in a general form
\begin{equation}
T\left(x,y\right)=\sum_{\gamma\in\mathbb{Z}}\sum_{n\in\mathbb{Q}_{p}/Z_{p}}T^{\left(\gamma,n\right)}\delta\left(|x-y|_{p}-p^{\gamma}\right)\Omega\left(|x-np^{-\gamma}|_{p}p^{-\gamma}\right).\label{Ker_T}
\end{equation}
In \cite{Koz,Koz1} it is also shown that wavelets (\ref{V-B}) are eigenfunctions of operator (\ref{T}) with kernel~(\ref{Ker_T}) and eigenvalues
\begin{equation}
\lambda_{\gamma n}=-p^{\gamma}T^{\left(\gamma,n\right)}-\left(1-p^{-1}\right)\sum_{\gamma'=\gamma+1}^{\infty}p^{\gamma'}T^{\left(\gamma',p^{\gamma'-\gamma}n\right)}.\label{EV_K}
\end{equation}
From formulas (\ref{f-psi}) and (\ref{phi-psi}), which express the relation of the overcomplete
system of functions $f_{\gamma,n,a}\left(x\right)$ and the orthonormal
basis $\varphi_{\gamma,n,b}\left(x\right)$ in terms of the functions of the wavelet basis  $\psi_{\gamma,n,j}\left(x\right)$, it follows that functions
(\ref{f_i_n_a_Q_p}) and (\ref{bas}) are also eigenfunctions of the pseudo-differential operator~(\ref{T}) with kernel~(\ref{Ker_T}) and eigenvalues~(\ref{EV_K}).

To demonstrate the application of basis (\ref{bas}) we consider the Cauchy problem for the equation
\begin{equation}
\dfrac{d}{dt}f\left(x,t\right)=Tf\left(x\right)\label{KP_T}
\end{equation}
with the pseudo-differential operator~(\ref{T}) with kernel~(\ref{Ker_T}) and the initial condition
\[
f\left(x,0\right)=\Omega\left(|x|_{p}\right).
\]

A formal solution of (\ref{KP_T}) may be written as
\[
f\left(x,t\right)=\exp\left(Tt\right)f\left(x,0\right).
\]
Hence, using expansion (\ref{omega_dec}), one may write this solution as a series
\[
f\left(x,t\right)=\exp\left(Tt\right)\left(\sum_{\gamma=1}^{\infty}p^{-\gamma+1}f_{\gamma,0,0}\left(x\right)\right)=
\]
\begin{equation}
=\sum_{\gamma=1}^{\infty}p^{-\gamma+1}\left(\Omega\left(|x|_{p}p^{-\gamma+1}\right)-p^{-1}\Omega\left(|x|_{p}p^{-\gamma}\right)\right)\exp\left(-\left(p^{\gamma}T^{\left(\gamma,n\right)}+\left(1-p^{-1}\right)\sum_{\gamma'=\gamma+1}^{\infty}p^{\gamma'}T^{\left(\gamma',p^{\gamma'-\gamma}n\right)}\right)t\right).\label{solution-1}
\end{equation}

\section{Conclusions}

In the present paper we construct new real bases formed by
eigenfunctions of the Vladimirov operator on a~compact set  $B_{r}\subset \mathbb{Q}_{p}$
and on the field $\mathbb{Q}_{p}$. It should be noted that all the previously available bases of the Vladimirov operator (\cite{Koch,V1,V2,V3,Koz,Koz1,KK}) were expressed
in terms of the characters of the field $\mathbb{Q}_{p}$ as an additive group.
A~similar construction can be used only on ultrametric spaces which feature a~group structure. Such spaces include the field  $\mathbb{Q}_{p}$ of $p$-adic numbers,
the ring $\mathbb{Q}_{m}$ of  $m$-adic numbers \cite{DZ}, the ring
of $a$-adic numbers~\cite{HR}. All these spaces are homogeneous ultrametric spaces пространствами (that is, spaces, in which for any fixed
$r$~there exists a~number $N\left(r\right)$ such that
the ball $B_{r}(a)$ can be represented as a~union of  $N\left(r\right)$
balls of radii $r'<r$). The basis introduced in the present paper does not contain characters of the field $\mathbb{Q}_{p}$ as an additive group, and can be
entirely expressed in terms of the characteristic functions of  balls in $\mathbb{Q}_{p}$.
By this reason such a~construction can be used to construct bases for inhomogeneous infinite ultrametric spaces without group structure.
Similar bases on finite ultrametric spaces with arbitrary topology were studied in \cite{BG,Motyl}.
In this connection, such types of bases can be the starting point in solving the spectral problem for the analogues
of the Vladimirov operator in an ultrametric space not featuring a~group structure by reducing this problem to a~spectral problem of
some modification of the Vladimirov operator  in~$\mathbb{Q}_{p}$. The last problem will be addressed in a~forthcoming  paper~\cite{Z} by the authors.

\vspace{5mm}
{\bf Acknowledgements}

The authors are  grateful to Prof.\ A.\,R.~Alimov (Faculty of Mechanics and Mathematics,
Moscow State University) for his assistance in preparing the manuscript  and a~number of helpful comments.

\smallskip

The work was partially supported by the Ministry of Education and Science of of the Russian Federation under the Competitiveness Enhancement Program of SSAU for 2013--2020.


\begin{thebibliography}{10}
\bibitem{V} Vladimirov V. S., Generalized functions over the field
of $p$-adic numbers // Russian Math. Surveys. --1988. -- Vol 43, no.
5. -- P. 19--64.

\bibitem{VVZ} Vladimirov V. S., Volovich I. V., Zelenov E. I. $p$-Adic
Analysis and Mathematical Physics. -- Singapure: World Scientific
Publishing, 1994.

\bibitem{DKKV} Dragovich B., Khrennikov A. Yu., Kozyrev S. V. , Volovich
I. V. On $p$-adic mathematical physics // \emph{$p$}-Adic Numbers,
Ultrametric Analysis, and Applications. -- 2009. -- Vol. 1., no. 1.
-- P. 1--17.

\bibitem{S-Sh-1} Smolyanov O. G., Shamarov N. N. Feynman Formulas
and Path Integrals for Evolution Equations with the Vladimirov Operator
// Proceedings of the Steklov Institute of Mathematics. -- 2009. --
265. -- P. 217--228.

\bibitem{S-Sh-2} Smolyanov O. G., Shamarov N. N. Hamiltonian Feynman
formulas for equations containing the vladimirov operator with variable
coefficients // Dokl. Math. -- 2011. -- Vol. 84, no. 2. -- P. 689--694.

\bibitem{Koch} Kochubei A. N. Pseudodifferential equations and stochastics
over non-archimedtan fields. -- New York: Marcel Dekker, 2001.

\bibitem{S-H} Albeverio S., Khrennikov A. Yu., and Shelkovich V.
M. Theory of $p$-adic Distributions: Linear and Nonlinear Models.
-- Cambridge University Press, 2010.

\bibitem{ABKO} Avetisov V. A., Bikulov A. Kh., Kozyrev S. V., Osipov
V. A. $p$-Adic Models of Ultrametric Diffusion Constrained by Hierarchical
Energy Landscapes // J. Phys. A: Math. Gen. -- 2002. -- Vol. 35. --
P. 177--189.

\bibitem{ABO} Avetisov V. A., Bikulov A. Kh., Osipov V. A. $p$-Adic
Description of Characteristic Relaxation in Complex Systems // J.
Phys. A: Math. Gen. --2003. -- Vol. 36. -- P. 4239--4246.

\bibitem{ABZ} Avetisov V. A., Bikulov A. Kh., Zubarev A. P. First
Passage Time Distribution and the Number of Returns for Ultrametric
Random Walks // J. Phys. A: Math. Theor. -- 2009. -- Vol. 42. -- P.
085003 -- 085020.

\bibitem{ABZ_MIAN} Avetisov V. A., Bikulov A. Kh., Zubarev A. P.
Ultrametric random walk and dynamics of protein molecules // Proceedings
of the Steklov Institute of Mathematics. -- 2014. -- Vol. 285. --
P. 3--25.

\bibitem{ABZ_motor} Avetisov V. A., Bikulov A. Kh., Zubarev A. P.
Mathematical Modeling of Molecular ``nano-machines'' // Vestn. Samar.
Gos. Tekhn. Univ. Ser. Fiz.-Mat. Nauki. -- 2011. -- Vol. 1 (22). --
P. 9--15 (in Russian).

\bibitem{BZK} Bikulov A. Kh., Zubarev A. P., Kaidalova L. V. Hierarchic
Dynamic Model of Financial Market near Crash Points and $p$-adic
Mathematical Analysis // Vestn. Samar. Gos. Tekhn. Univ. -- 2006.
-- Vol. 42. -- P. 135--141 (in Russian).

\bibitem{V1} Vladimirov V. S. $p$-Adic analysis and $p$-adic quantum
mechanics. -- Ann. of the NY Ac. Sci.: Symposium in Frontiers of Math.,
1988.

\bibitem{V2} Vladimirov V. S. Ramified Characters of Idele Groups
of One-Class Quadratic Fields // Proc. Steklov Inst. Math. -- 1999.
--Vol. 224. -- P. 107--114.

\bibitem{V3} Vladimirov V. S. On the spectrum of some pseudodifferential
operators over the field of $p$-adic numbers // Leningrad Math. J.,
-- 1991. -- Vol. 2, no. 6. -- P. 1261--1278.

\bibitem{Koz} Kozyrev S. V. Wavelet theory as $p$-adic spectral
analysis // Izv. Math. -- 2002. -- 66:2. -- P. 367--376

\bibitem{Koz1} Kozyrev S. V. $p$-Adic Pseudodifferential Operators
and $p$-Adic Wavelets // Theoret. and Math. Phys. -- 2004. -- Vol.
138, no. 3. -- P. 322--332.

\bibitem{KK} Kozyrev S. V., Khrennikov A. Yu. Pseudodifferential
Operators on Ultrametric Spaces and Ultrametric Wavelets // Izv. RAN.
Ser. Mat. -- 2005. -- Vol. 69, no. 5. -- P. 133--148.

\bibitem{Bik_Vol} Bikulov A. Kh., Volovich I. V. $p$-adic Brownian
motion // Izv. RAN. Ser. Mat. -- 1997. -- Vol. 61, no. 3. -- P. 75--90.

\bibitem{DZ} Dolgopolov M. V., Zubarev A. P. Some Aspects of $m$-Adic
Analysis and Its Applications to $m$-Adic Stochastic Processes //
$p$-Adic Numbers, Ultrametric Analysis, and Applications. -- 2011.
-- Vol. 3. -- no. 1. -- P. 39--51.

\bibitem{HR} Hewitt E., Ross K.\,A., Abstract Harmonic Analysis. --
Springer-Verlag, 1987.

\bibitem{BG} Bachas C. P., Huberman B. A. Complexity and the Relaxation
of Hierarchical Structures // Phys. Rev. Lett. -- 1986. -- Vol. 57.
-- P. 1965--1969.

\bibitem{Motyl} Motyl W. G. Dynamics on random ultrametric spaces
// J. Phys. A: Math. Gen. -- 1987. -- Vol. 20. -- P. 5481--5488.


\bibitem{Z} Bikulov A. Kh. and Zubarev A. P. Application of $p$--Adic Analysis Methods in Describing Markov Processes on Ultrametric Spaces Isometrically Embedded into $Q_p$  // $p$--Adic Numbers, Ultrametric Analysis, and Applications. -- 2015. -- Vol. 7, no. 2. -- P. 111--122.
\end{thebibliography}
\end{document}